\documentclass[a4paper]{article}
\usepackage[margin=1in]{geometry}
\usepackage{authblk}
\usepackage[utf8]{inputenc}
\usepackage[T1]{fontenc}
\usepackage{microtype}
\usepackage{amsfonts}
\usepackage{amssymb}
\usepackage{amsthm}
\usepackage{amsmath}
\usepackage{comment}
\usepackage{times}
\usepackage{enumitem}
\usepackage[hidelinks]{hyperref}
\usepackage[capitalise]{cleveref}
\usepackage{thmtools}

\newtheorem{theorem}{Theorem}[section]
\newtheorem{lemma}[theorem]{Lemma}
\newtheorem{corollary}[theorem]{Corollary}

\newcommand{\rank}{\text{\rm rank}}
\newcommand{\select}{\text{\rm select}}
\newcommand{\lcp}{\text{\rm lcp}}
\newcommand{\lcs}{\text{\rm lcs}}
\newcommand{\T}{T}
\newcommand{\C}{C}
\newcommand{\LCP}{\text{\rm LCP}}
\newcommand{\LCS}{\text{\rm LCS}}
\newcommand{\PLCP}{\text{\rm PLCP}}
\newcommand{\LF}{\text{\rm LF}}
\newcommand{\PSI}{\text{\rm $\Psi$}}
\newcommand{\SA}{\text{\rm SA}}
\newcommand{\ISA}{\text{\rm ISA}}
\newcommand{\BWT}{\text{\rm BWT}}

\newcommand{\SPLCP}{\mathrm{PLCP_{succ}}}
\newcommand{\MATRIX}{\mathcal{M}}
\newcommand{\NSV}{\text{\rm NSV}}
\newcommand{\PSV}{\text{\rm PSV}}
\newcommand{\LPF}{\text{\rm LPF}}
\newcommand\PSVL{\mbox{$\text{\rm PSV}_{\text{\rm lex}}$}}
\newcommand\PSVT{\mbox{$\text{\rm PSV}_{\text{\rm text}}$}}
\newcommand\NSVL{\mbox{$\text{\rm NSV}_{\text{\rm lex}}$}}
\newcommand\NSVT{\mbox{$\text{\rm NSV}_{\text{\rm text}}$}}
\newcommand\bigO{\mathcal{O}}

\let\OLDthebibliography\thebibliography                                         
\renewcommand\thebibliography[1]{                                               
  \OLDthebibliography{#1}                                                       
  \setlength{\parskip}{0pt}                                                     
  \setlength{\itemsep}{0pt plus 0.3ex}                                          
}

\begin{document}

\setenumerate{itemsep=0ex, parsep=1pt, topsep=1pt}

\renewcommand\Affilfont{\normalsize}

\title{Optimal Construction of Compressed Indexes\\ for Highly
    Repetitive Texts\thanks{Research partially supported by the Centre
    for Discrete Mathematics and its Applications (DIMAP) and by EPSRC
    award EP/N011163/1.}}
\author[1]{Dominik Kempa}
\affil[1]{Department of Computer Science and Helsinki Institute for Information Technology (HIIT), University of
    Helsinki, and Department of Computer Science and Centre for
    Discrete Mathematics and its Applications (DIMAP), University of
    Warwick.}
\affil[ ]{\href{mailto:dominik.kempa@warwick.ac.uk}{\nolinkurl{dominik.kempa@warwick.ac.uk}}}

\date{\vspace{-1.5cm}}
\maketitle

\begin{abstract}
  We propose algorithms that, given the input string of length $n$
  over integer alphabet of size $\sigma$, construct the
  Burrows--Wheeler transform (BWT), the permuted longest-common-prefix
  (PLCP) array, and the LZ77 parsing in
  $\bigO(n/\log_{\sigma}n+r\,{\rm polylog}\,n)$ time and working
  space, where $r$ is the number of runs in the BWT of the input.
  These are the essential components of many compressed indexes such
  as compressed suffix tree, FM-index, and grammar and LZ77-based
  indexes, but also find numerous applications in sequence analysis
  and data compression.  The value of $r$ is a common measure of
  repetitiveness that is significantly smaller than $n$ if the string
  is highly repetitive.  Since just accessing every symbol of the
  string requires $\Omega(n/\log_{\sigma}n)$ time, the presented
  algorithms are time and space optimal for inputs satisfying the
  assumption $n/r\in\Omega({\rm polylog}\,n)$ on the repetitiveness.
  For such inputs our result improves upon the currently fastest
  general algorithms of Belazzougui (STOC 2014) and Munro et al. (SODA
  2017) which run in $\bigO(n)$ time and use $\bigO(n/\log_{\sigma}
  n)$ working space.  We also show how to use our techniques to obtain
  optimal solutions on highly repetitive data for other fundamental
  string processing problems such as: Lyndon factorization,
  construction of run-length compressed suffix arrays, and some
  classical ``textbook'' problems such as computing the longest
  substring occurring at least some fixed number of times.
\end{abstract}

\section{Introduction}

The problem of text indexing is to preprocess the input text $T$ so
that given any query pattern $P$, we can quickly (typically
$\bigO(|P|+{\rm occ})$, where $|P|$ is the length of $P$ and ${\rm
  occ}$ is the number of occurrences of $P$ in $T$) find all
occurrences of $P$ in $T$. The two classical data structures for this
problem are the suffix tree~\cite{weiner} and the suffix
array~\cite{mm1993}. The suffix tree is a trie containing all suffixes
of $T$ with each unary path compressed into a single edge labeled by
the text substring. The suffix array is a list of suffixes of $T$ in
lexicographic order where each suffix is encoded using its starting
position. Both data structures take $\Theta(n)$ words of space. In
addition to indexing, these data structures underpin dozens of
applications in bioinformatics, data compression, and information
retrieval. Suffix arrays, in particular, have become central to modern
genomics, where they are used for genome assembly and short read
alignment, data-intensive tasks at the forefront of modern medical and
evolutionary biology~\cite{MBCT2015}. This can be attributed mostly to
their space-efficiency and simplicity.

In modern applications, however, which require indexing datasets of
size close to the size of available RAM, even the suffix arrays can be
prohibitively large, particularly in applications where the text
consists of symbols from some alphabet $\Sigma$ of small size
$\sigma=|\Sigma|$ (e.g., in bioinformatics $\Sigma=\{{\tt A}, {\tt C},
{\tt G}, {\tt T}\}$ and so $\sigma=4$).  For such collections, the
classical indexes are $\Theta(\log_{\sigma}n)$ times larger than the
text which requires only $\Theta(n \log \sigma)$ bits, i.e.,
$\Theta(n/\log_{\sigma}n)$ words.

The invention of FM-index~\cite{FM,fm2005} and the compressed suffix
array (CSA)~\cite{CSA,GrossiV05} at the turn of the millennium
addressed this issue and revolutionized the field of string algorithms
for nearly two decades. These data structures require only
$\bigO(n/\log_{\sigma}n)$ words of space and provide random access to
the suffix array in $\bigO(\log^{\epsilon}n)$ time. Dozens of papers
followed the two seminal papers, proposing various improvements,
generalizations, and practical implementations
(see~\cite{nm2007,Navarro14,BelazzouguiN14} for excellent
surveys). These indexes are now widespread, both in theory where they
provide off-the-shelf small space indexing structures and in
practice, particularly bioinformatics, where they are the central
component of many read-aligners~\cite{bowtie,bwa}.

The other approach to indexing, recently gaining popularity due to
the quick increase in the amount of highly repetitive data, such as
software repositories or genomic databases is designing indexes
specialized for repetitive strings. The first such
index~\cite{LZ-index} was based on the Lempel--Ziv (LZ77)
parsing~\cite{LZ77}, the popular dictionary compression algorithms
(used, e.g., in gzip and 7-zip compressors). Many improvements to the
basic scheme were proposed since
then~\cite{GagieGKNP14,BilleEGV18,DCC2015,GagieGKNP12,ArroyueloNS12,ArroyueloN11,ArroyueloN10},
and now the performance of LZ-based indexes is often on par with the
FM-index or CSA~\cite{FerradaKP18}.  Independently to the development
of LZ-based indexes, it was observed that the Burrows--Wheeler
transform (BWT)~\cite{bw1994}, which underlies the FM-index, produces
long runs of characters when applied to highly repetitive
data~\cite{MakinenNSV10,SirenVMN08}.  Gagie et al.~\cite{GagieNP17}
recently proposed a run length compressed suffix array (RLCSA) that
provides fast access to suffix array and pattern matching queries in
$\bigO(r\,{\rm polylog}\,n)$ or even $\bigO(r)$ space, where $r$ is
the number of runs in the BWT of the text. The value of $r$ is, next
to $z$ (the size of LZ77 parsing), a common measure of
repetitiveness~\cite{attractors}.

Given the small space usage of the compressed indexes, their
space-efficient construction emerged as one of the major open
problems. A gradual improvement~\cite{LamSSY02,HonLSS03,HonSS03} in
the construction of compressed suffix array culminated with the work
of Belazzougui~\cite{STOC2014} who described the (randomized)
$\bigO(n)$ time construction working in optimal space of
$\bigO(n/\log_{\sigma}n)$. An alternative (and deterministic)
construction was proposed by Munro et al.~\cite{SODA2017}.  These
algorithms achieve the optimal construction space but their running
time is up to $\Theta(\log n)$ times larger than the lower bound of
$\Omega(n/\log_{\sigma}n)$ time (required to read the input/write the
output).

\paragraph*{Our Contribution}

We propose algorithms that, given the input string of length $n$ over
integer alphabet of size $\sigma$, construct the Burrows--Wheeler
transform (BWT), the permuted longest-common-prefix (PLCP) array, and
the LZ77 parsing in $\bigO(n/\log_{\sigma}n+r\,{\rm polylog}\,n)$ time
and working space, where $r$ is the number of runs in the BWT of the
input.

These are the essential components of nearly every compressed text
index developed in the last two decades: all variants of FM-index rely
on BWT~\cite{FM,GagieNP17}, compressed suffix arrays/trees rely on
$\Psi$~\cite{CSA,Sada02} (which is dual to the
$\BWT$~\cite{SODA2017,HonSS03}) and the $\PLCP$ array, and LZ77-based
and grammar-based indexes rely on the LZ77
parsing~\cite{LZ-index,Rytter03}.  Apart from text indexing, these
data structures have also numerous applications in sequence analysis
and data compression~\cite{MBCT2015,navarrobook,ohl2013}.

Since just accessing every symbol of the string requires
$\Omega(n/\log_{\sigma}n)$ time, the presented algorithms are time and
space optimal for inputs satisfying the assumption $n/r\in\Omega({\rm
  polylog}\,n)$ on the repetitiveness.  Our results have particularly
important implications for bioinformatics, where most of the data is
highly-repetitive~\cite{MakinenNSV10,MBCT2015,MakinenNSV09} and over
small (DNA) alphabet.  For such inputs, our result improves upon the
currently fastest general algorithms of Belazzougui~\cite{STOC2014}
and Munro et al.~\cite{SODA2017} which run in $\bigO(n)$ time and use
$\bigO(n/\log_{\sigma} n)$ working space.

We also show how to use our techniques to obtain an
$\bigO(n/\log_{\sigma}n+r\,{\rm polylog}\,n)$ time and space
algorithms for other fundamental string processing problems such as:
Lyndon factorization~\cite{CFL58}, construction of run-length
compressed suffix arrays~\cite{GagieNP17}, and some classical
``textbook'' problems such as computing the longest substring
occurring at least some fixed number of times.

On the way to the above results, we show how to generalize the RLCSA of
Gagie et al.~\cite{GagieNP17} to achieve a trade-off between index
size and query time. In particular, we obtain a $\bigO(r\,{\rm
  polylog}\,n)$-space data structure that can answer suffix array
queries in $\bigO(\log n / \log \log n)$ time which improves on the
$\bigO(\log n)$ query time of~\cite{GagieNP17}.

\section{Preliminaries}

We assume a word-RAM model with a word of $w=\Theta(\log n)$ bits and
with all usual arithmetic and logic operations taking constant
time. Unless explicitly specified otherwise, all space complexities
are given in words. All our algorithms are deterministic.

Throughout we consider a string $\T[1..n]$ of symbols from an alphabet
$\Sigma=[1..\sigma]$ of size $\sigma\,{\leq}\, n$. We assume
$T[n]\,{=}\,\$$ with a numerical value of $\$$ equal to 0.  For $j \in
[1..n]$, we write $\T[j..n]$ to denote the suffix $j$ of $\T$. We
define the \emph{rotation} of $\T$ as a string $\T[j..n]\T[1..j-1]$
for any position $j\in [1..n]$.

The \emph{suffix array}~\cite{mm1993,gbys1992} of $\T$ is an array
$\SA[1..n]$ which contains a permutation of the integers $[1..n]$ such
that $\T[\SA[1]..n] \prec \T[\SA[2]..n] \prec \cdots \prec
\T[\SA[n]..n]$, where $\prec$ denotes the lexicographical order.  The
inverse suffix array $\ISA$ is the inverse permutation of $\SA$, i.e.,
$\ISA[j] = i$ iff $\SA[i] = j$. The array $\Phi[1..n]$
(see~\cite{KarkkainenMP2009}) is defined by $\Phi[\SA[i]]=\SA[i-1]$
for $i\in[2..n]$, and $\Phi[\SA[1]]=\SA[n]$, that is, the suffix
$\Phi[j]$ is the immediate lexicographical predecessor of suffix $j$.

Let $\lcp(j_1,j_2)$ denote the length of the longest-common-prefix
(LCP) of suffix $j_1$ and suffix $j_2$.  The
\emph{longest-common-prefix array}~\cite{mm1993,klaap2001},
$\LCP[1..n]$, is defined as $\LCP[i] = \lcp(\SA[i],\SA[i-1])$ for $i
\in [2..n]$ and $\LCP[1]=0$.  The \emph{permuted LCP
  array}~\cite{KarkkainenMP2009} $\PLCP[1..n]$ is the LCP array
permuted from the lexicographical order into the text order, i.e.,
$\PLCP[\SA[i]]=\LCP[i]$ for $i \in [1..n]$. Then $\PLCP[j] =
\lcp(j,\Phi[j])$ for all $j\in[1..n]$.

The \emph{succinct PLCP array}~\cite{Sada02,KarkkainenMP2009}
$\SPLCP[1..2n]$ represents the PLCP array using $2n$
bits. Specifically, $\SPLCP[j']=1$ if $j'=2j+\PLCP[j]$ for some
$j\in[1..n]$, and $\SPLCP[j']=0$ otherwise.  Any lcp value can be
recovered by the equation $\PLCP[j]=\select_{\SPLCP}(1,j)-2j$, where
$\select_S(c,j)$ returns the location of the $j^{\mbox{\scriptsize
    th}}$ $c$ in $S$.

The \emph{Burrows--Wheeler transform}~\cite{bw1994} $\BWT[1..n]$ of
$\T$ is defined by $\BWT[i] = \T[\SA[i]-1]$ if $\SA[i] > 1$ and
$\BWT[i] = \T[n]$ otherwise. Let $\MATRIX$ denote the $n \times n$
matrix, whose rows are lexicographically sorted rotations of $\T$. We
denote the rows by $\MATRIX[i]$, $i\in [1..n]$. Note that $\BWT$ is
the last column of $\MATRIX$.

The \emph{LF-mapping}~\cite{FM} is defined by the equation
$\LF[\ISA[j]] = \ISA[j-1]$, $j\in[2..n]$, and $\LF[\ISA[1]] =
\ISA[n]$.  By $\PSI$ we denote the inverse of $\LF$.  The significance
of $\LF$ (and the principle underlying FM-index~\cite{FM}) lies in the
fact that, for $i\in[1..,n]$, $\LF[i] = \C[\BWT[i]]\,{+}\,
\rank_{\BWT}(\BWT[i],i)$, where $\C[c]$ is the number of symbols in
$\T$ that are smaller than $c$, and $\rank_S(c,i)$ is the number of
occurrences of $c$ in $S[1..i]$.  From the formula for $\LF$ we obtain
the following fact.

\begin{lemma}
  \label{lm:lf-in-run}
  Let $\BWT[b..e]$ be a run of the same symbol and let
  $i,i'\in[b,e]$. Then, $\LF[i]=\LF[i']+(i-i')$.
\end{lemma}

If $i$ is the rank (i.e., the number of smaller suffixes) of $P$ among
suffixes of $\T$, then $\C[c]+\rank_{\BWT}(c,i)$ is the rank of
$cP$. This is called \emph{backward search}~\cite{FM}.

We say that an lcp value $\LCP[i]=\PLCP[\SA[i]]$ is \emph{reducible}
if $\BWT[i] = \BWT[i-1]$ and \emph{irreducible} otherwise. The
significance of reducibility is summarized in the following two
lemmas.

\begin{lemma}[\cite{KarkkainenMP2009}]
  \label{lm:reducible}
  If $\PLCP[j]$ is reducible, then $\PLCP[j]=\PLCP[j-1]-1$ and
  $\Phi[j]=\Phi[j-1]+1$.
\end{lemma}

\begin{lemma}[\cite{KarkkainenMP2009,KarkkainenKP16}]
  \label{lm:irreducible}
  The sum of all irreducible lcp values is $\le n \log n$.
\end{lemma}

It can be shown~\cite{MakinenNSV10} that repetitions in $T$ generate
equal-letter runs in BWT.  By $r$ we denote the number of runs in
$\BWT$.  We can efficiently represent this transform as the list of
pairs ${\rm RLBWT} = \langle \lambda_i, c_i \rangle_{i=1,\dots, r}$,
where $\lambda_i>0$ is the starting position of the $i$-th run and
$c_i\in\Sigma$.  Note that $r$ is also the number of irreducible lcp
values.

\section{Augmenting RLBWT}

In this section we present extensions of run-length compressed BWT
needed by our algorithms. Each extension expands its functionality
while maintaining small space usage and low construction time/space.

\subsection{Rank and Select Support}

One of the basic operations we will need are rank and select queries
on BWT. We will now show that a run-length compressed BWT can be
quickly augmented with a data structure capable of answering these
queries in BWT-runs space.

\begin{theorem}
  \label{thm:ranksel-support}
  Given RLBWT of size $r$ for text $\T[1..n]$ we can add $\bigO(r)$
  space so that, given $i\in[0..n]$ and $c\in[1..\sigma]$, values
  $\rank_{\BWT}(c,i)$ and $\select_{\BWT}(c,i)$ can be computed in
  $\bigO(\log r)$ time. The data structure can be constructed in
  $\bigO(r \log r)$ time using $\bigO(r)$ space.
\end{theorem}

\begin{proof}
  We augment each BWT-run with its length and sort the runs using the
  symbol as the primary key, and the start of the run as the secondary
  key. This allows us to compute, for every run $[b..e]$, the value
  $\rank_{\BWT}(c,b)$ where $c=\BWT[b]$. Using this list, both queries
  can be answered in $\bigO(\log r)$ time using binary search.
\end{proof}

\subsection{\texorpdfstring{$\LF/\PSI$}{LF/PSI} and Backward Search Support}

We now show that with the help of the above rank/select data
structures we can support more complicated navigational queries,
namely, given any $i\in [1..n]$ such that $\SA[i]=j$ we can compute
$\ISA[j-1]$ (i.e., $\LF[i]$) and $\ISA[j+1]$ (i.e., $\PSI[i]$). Note
that none of the queries will require the knowledge of $j$. As a
simple corollary, we obtain efficient support for backward search on
RLBWT.

\begin{theorem}
  \label{thm:lfpsi-support}
  Given RLBWT of size $r$ for text $\T[1..n]$ we can add $\bigO(r)$
  space so that, given $i\in[1..n]$, values $\LF[i]$ and $\PSI[i]$ can
  be computed in $\bigO(\log r)$ time. The data structure can be
  constructed in $\bigO(r \log r)$ time using $\bigO(r)$ working
  space.
\end{theorem}

\begin{proof}
  Similarly as in Theorem~\ref{thm:ranksel-support} we prepare a
  (sorted) list containing, for each symbol $c$ occurring in $\T$, the
  total frequency of symbols smaller than $c$.

  To answer $\LF[i]$ we first compute $\BWT[i]$ (by searching the list
  of runs), then $C[\BWT[i]]$ (by searching the above frequency
  table), and finally apply Theorem~\ref{thm:ranksel-support}.  To
  compute $\Psi[i]$ we first determine (using the frequency table) the
  symbol $c$ following $\BWT[i]$ in text and the number $k$ such that
  this $c$ is the $k$-th occurrence of $c$ in the first column of
  $\MATRIX$. It then remains to find the $k$-th occurrence of $c$ in
  the BWT using Theorem~\ref{thm:ranksel-support}.
\end{proof}

\begin{corollary}
  \label{thm:bs-support}
  Given RLBWT of size $r$ for text $\T[1..n]$ we can add $\bigO(r)$
  space so that, given a rank $i\in[0..n]$ of a string $P$ among the
  suffixes of $T$, for any $c\in [1..\sigma]$ we can compute in
  $\bigO(\log r)$ time the rank of $cP$. The data structure can be
  constructed in $\bigO(r \log r)$ time using $\bigO(r)$ working
  space.
\end{corollary}

\subsection{Suffix-Rank Support}
\label{sec:suffix-rank-support}

In this section we describe an extension of RLBWT that will allow us
to efficiently merge two RLBWTs during the BWT construction
algorithm. We start by defining a generalization of BWT-runs and
stating their basic properties.

Let $\lcs(x,y)$ denote the length of the longest common suffix of
strings $x$ and $y$. We define the $\LCS[1..n]$
array~\cite{KarkkainenKP12} as $\LCS[i]=\lcs(\MATRIX[i],
\MATRIX[i-1])$ for $i\in [2..n]$ and $\LCS[1]=0$ (recall that
$\mathcal{M}$ is a matrix containing sorted rotations of $T$).  Let
$\tau \geq 1$ be an integer. We say that a range $[b..e]$ of BWT is a
\emph{$\tau$-run} if $\LCS[b]<\tau$, $\LCS[e+1]<\tau$, and for any
$i\in[b+1..e]$, $\LCS[i]\geq \tau$. By this definition, a BWT run is a
1-run.  For $j\geq 0$ let $Q_{j} = \{i\in[1..n] \mid \LCS[i]=j\}$ and
$R_{\tau}=\bigcup_{j=0}^{\tau-1}Q_j$.  Then, $R_{\tau}$ is exactly the
set of starting positions of $\tau$-runs.

\begin{lemma}[\cite{KarkkainenKP12}]
  For any $i\in [2..n]$,
  \[
    \LCS[i] = \left\{
      \begin{array}{l l}
        0 & \enspace \text{{\rm if}
          $\BWT[i]\neq\BWT[i-1]$},\\ \LCS[\LF[i]]+1 & \enspace
        \text{{\rm otherwise.}}\\
      \end{array}
    \right.
   \]
\end{lemma}

Since $\PSI$ is the inverse of $\LF$ we obtain that for any $j\geq 1$,
$Q_{j}=\{\PSI[i]\ |\ i\in Q_{j-1}\ {\rm and}\ \PSI[i]\notin Q_0\}$.
Thus, the set $R_{\tau}$ can be efficiently computed by iterating each
of the starting positions of BWT-runs $\tau-1$ times using $\PSI$ and
taking a union of all visited positions.  From the above we see that
$|Q_{j+1}| \leq |Q_j|$, which implies that the number of $\tau$-runs
satisfies $|R_{\tau}| \leq |Q_0|\tau=r\tau$.

\begin{theorem}
  \label{thm:rank-support}
  Let $S[1..m]$, $S'[1..m']$ be strings with $r$ and $r'$
  (respectively) runs in the $\BWT$. Given RLBWTs of $S$ and $S'$ it
  is possible, for any integer $\tau \geq 1$, to build a data
  structure of size $\bigO(\frac{m}{\tau} + r+r')$ that can, given a
  rank $i\in [0..m]$ of some suffix $S[j..m]$ among suffixes of $S$,
  compute the rank of $S[j..m]$ among suffixes of $S'$ in
  $\bigO(\tau(\log\frac{m}{\tau}+\log r + \log r'))$ time. The data
  structure can be constructed in $\bigO(\tau^2(r+r') \log (r\tau +
  r'\tau) + \frac{m}{\tau}(\log (r\tau) + \log (r'\tau) + \log
  \frac{m}{\tau}))$ time and $\bigO(\tau^2(r+r') + \frac{m}{\tau})$
  space.
\end{theorem}

\begin{proof}
  We start by augmenting both RLBWTs with $\PSI$ and $\LF$ support
  (Theorem~\ref{thm:lfpsi-support}) and RLBWT of $S'$ with the
  backward search support (Corollary~\ref{thm:bs-support}). This
  requires $\bigO(r \log r + r' \log r')$ time and $\bigO(r+r')$
  space.

  We then compute a (sorted) set of starting positions of $\tau$-runs
  for both RLBWTs. For $S$ this requires answering $r\tau$
  $\PSI$-queries which takes $\bigO(r\tau \log r)$ time in total, and
  then sorting the resulting set of positions in $\bigO((r\tau) \log
  (r\tau))$ time. Analogous processing for $S'$ takes $\bigO((r'\tau)
  \log (r'\tau))$ time. The starting positions of all $\tau$-runs
  require $\bigO((r+r')\tau)$ space in total.

  Next, for any $\tau$-run $[b..e]$ we compute and store the
  associated $\tau$ symbols. We also store the value $\LF^{\tau}[b]$,
  but only for $\tau$-runs of $S$. Due to simple generalization of
  Lemma~\ref{lm:lf-in-run}, this will allow us to compute the value
  $\LF^{\tau}[i]$ for \emph{any} $i$. In total this requires answering
  $\tau^2(r+r')$ $\LF$-queries and hence takes $\bigO(\tau^2(r+r')
  \log r)$ time. The space needed to store all symbols is
  $\bigO(\tau^2(r+r'))$.

  We then lexicographically sort all length-$\tau$ strings associated
  with $\tau$-runs (henceforth called \emph{$\tau$-substrings}) and
  assign to each run the rank of the associated substring in the
  sorted order. Importantly, $\tau$-substrings of $S$ and $S'$ are
  sorted together. These ranks will be used as order-preserving names
  for $\tau$-substrings. We use an LSD string sort with a stable
  comparison-based sort for each position hence the sorting takes
  $\bigO\left(\tau^2(r+r') \log (r\tau+r'\tau)\right)$ time. The
  working space does not exceed $\bigO(\tau(r+r'))$. After the names
  are computed, we discard the substrings.

  We now observe that order-preserving names for $\tau$-substrings
  allow us to perform backward search $\tau$ symbols at a time. We
  build a rank-support data structure analogous to the one from
  Theorem~\ref{thm:ranksel-support} for names of $\tau$-substrings of
  $S'$.  We also add support for computing the total number of
  occurrences of names smaller than a given name. This takes
  $\bigO(r'\tau \log (r'\tau))$ time and $\bigO(r'\tau)$ space.  Then,
  given a rank $i$ of suffix $S[j..m]$ among suffixes of $S'$, we can
  compute the rank of suffix $S[j-\tau..m]$ among suffixes of $S'$ in
  $\bigO(\log (r'\tau))$ time by backward search on $S'$ using $i$ as
  a position, and the name of $\tau$-substring preceding $S[j..m]$ as
  a symbol.

  We now use the above multi-symbol backward search to compute the
  rank of every suffix of the form $S[m-k\tau..m]$ among suffixes of
  $S'$. We start from the shortest suffix and increase the length by
  $\tau$ in every step.  During the computation we also maintain the
  rank of the current suffix of $S$ among suffixes of $S$. This allows
  us to efficiently compute the name of the preceding
  $\tau$-substring. The rank can be updated using values $\LF^{\tau}$
  stored with each $\tau$-run of $S$. Thus, for each of the $m/\tau$
  suffixes of $S$ we obtain a pair of integers ($i_S$, $i_{S'}$),
  denoting its rank among the suffixes of $S$ and $S'$. We store these
  pairs as a list sorted by $i_S$. Computing the list takes
  $\bigO\left(\frac{m}{\tau}(\log (r\tau)+\log(r'\tau)) +
  \frac{m}{\tau}\log \frac{m}{\tau}\right)$ time.  After the list is
  computed we discard all data structures associated with $\tau$-runs.

  Using the above list of ranks we can answer the query from the claim
  as follows.  Starting with $i$, we compute a sequence of $\tau$
  positions in the BWT of $S$ by iterating $\PSI$ on $i$. For each
  position we can check in $\bigO(\log \frac{m}{\tau})$ time whether
  that position is in the list of ranks. Since we evenly sampled text
  positions, one of these positions has to correspond to the suffix of
  $S$ for which we computed the rank in the previous step.  Suppose we
  found such position after $\Delta\leq \tau$ steps, i.e., we now have
  a pair ($i_S$, $i_{S'}$) such that $i_{S'}$ is the rank of
  $S[j+\Delta..m]$ among suffixes of $S'$. We then perform $\Delta$
  steps of the standard backward search starting from rank $i_{S'}$ in
  the BWT of $S'$ using symbols $S[j{+}\Delta{-}1]$, \ldots,
  $S[j]$. This takes $\bigO\left(\Delta (\log r + \log
  r')\right)=\bigO\left(\tau(\log r + \log r')\right)$ time.
\end{proof}

\section{Construction of BWT}

In this section we show that given the packed encoding of text
$\T[1..n]$ over alphabet $\Sigma=[1..\sigma]$ of size $\sigma\leq n$
(i.e., using $\bigO(n / \log_{\sigma}n)$ words of space), we can
compute the packed encoding of $\BWT$ of $\T$ in
$\bigO(n/\log_{\sigma}n+r \log^7 n)$ time and
$\bigO(n/\log_{\sigma}n+r \log^5 n)$ space, where $r$ is the number of
runs in the $\BWT$ of $\T$.

\subsection{Algorithm Overview}

The basic scheme of our algorithm follows the algorithm of Hon et
al.~\cite{HonSS03}. Assume for simplicity that $w/\log {\sigma}=2^k$
for some integer $k$.  The algorithm works in $k+1$ rounds, where
$k=\log \log_{\sigma}n$.  In the $i$-th round, $i\in[0..k]$, we
interpret $\T$ as a string over superalphabet
$\Sigma_{i}=[1..\sigma_i]$ of size $\sigma_i=\sigma^{2^i}$, i.e., we
group symbols of $T$ into supersymbols consisting of $2^i$ original
symbols. We denote this string as $\T_i$. The rounds are executed in
decreasing order of $i$.  The input to the $i$-th round,
$i\in[0..k{-}1]$, is the run-length compressed BWT of $\T_{i+1}$, and
the output is the run-length compressed BWT of $\T_i$.  We denote the
size of RLBWT of $T_i$ by $r_i$.  The final output is the run-length
compressed BWT of $\T_0=\T$, which we then convert into packed
encoding taking $\bigO(n / \log_{\sigma}n)$ words.

For the $k$-th round, we observe that $|\Sigma_k|=\Theta(n)$ and
$|T_k|=\Theta(n/\log_{\sigma}n)$ hence to compute BWT of $T_k$ it
suffices to first run any of the linear-time algorithms for
constructing the suffix array~\cite{ks2003,NongZC11,ka05,KimSPP03} for
$T_k$ and then naively compute the RLBWT from the suffix array. This
takes $\bigO(n/\log_{\sigma}n)$ time and space.

Let $S=T_i$ for some $i\in[0..k{-}1]$ and suppose we are given the
RLBWT of $T_{i+1}$. Let $S_o$ be the string of length $|S|/2$ created
by grouping together symbols $S[2j-1]S[2j]$ for all $j$, and let $S_e$
be the analogously constructed string for pairs
$S[2j]S[2j+1]$. Clearly we have $S_o=T_{i+1}$ (recall that we start
indexing from 1).  Furthermore, it is easy to see that the BWT of $S$
can be obtained by interleaving BWTs of $S_o$ and $S_e$, and
discarding (more significant) half of the bits in the encoding of each
symbol.

The construction of RLBWT for $S$ consists of two steps: (1) first we
compute the RLBWT of $S_e$ from RLBWT of $S_o$, and then (2) merge
RLBWTs of $S_o$ and $S_e$ into RLBWT of $S$.

\subsection{Computing BWT of \texorpdfstring{$S_e$}{Se}}

In this section we assume that $S=T_i$ for some $i\in[0..k{-}1]$ and
that we are given the RLBWT of $S_o=T_{i+1}$ of size $r_o=r_{i+1}$.
Denote the size of RLBWT of $S_e$ by $r_e$. We will show that RLBWT of
$S_e$ can be computed in $\bigO(r_e + r_o \log r_o)$ time using
$\bigO(r_o + r_e)$ working space.

Recall that both $S_o$ and $S_e$ are over alphabet $\Sigma_{i+1}$.
Each of the symbols in that alphabet can be interpreted as a
concatenation of two symbols in the alphabet $\Sigma_i$. Let $c$ be
the symbol of either $S_o$ or $S_e$ and assume that $c=S[j]S[j+1]$ for
some $j\in[1..|S|{-}1]$.  By \emph{major subsymbol} of $c$ we denote a
symbol (equal to $S[j]$) from $\Sigma_i$ encoded by the more
significant half of bits encoding $c$, and by \emph{minor subsymbol}
we denote symbol encoded by remaining bits (equal to $S[j+1]$).

We first observe that by enumerating all runs of the RLBWT of $S_o$ in
increasing order of their minor subsymbols (and in case of ties, in
the increasing order of run beginnings), we obtain (on the remaining
bits) the minor subsymbols of the $\BWT$ of $S_e$ in the correct
order. Such enumeration could easily be done in $\bigO(r_o \log r_o)$
time and $\bigO(r_o)$ working space. To obtain the missing (major)
part of the encoding of symbols in the BWT of $S_e$, it suffices to
perform the $\LF$-step for each of the runs in the BWT of $S_o$ in the
sorted order above (i.e., by minor subsymbol), and look up the minor
subsymbols in the resulting range of $\BWT$ of $S_o$.

The problem with the above approach is the running time. While it
indeed produces correct RLBWT of $S_e$, having to scan all runs in the
range of BWT of $S_o$ obtained by performing the $\LF$-step on each of
the runs of $S_o$ could be prohibitively high. To address this we
first construct a run-length compressed sequence of minor subsymbols
extracted from $\BWT$ of $S_o$ and use it to extract minor subsymbols
of $\BWT$ of $S_o$ in total time proportional to the number of runs in
the $\BWT$ of $S_e$.

\begin{lemma}
  \label{lm:inducing}
  Given RLBWT of size $r_o$ for $S_o=T_{i+1}$ we can compute the RLBWT
  of $S_e$ in $\bigO(r_e + r_o \log r_o)$ time and $\bigO(r_o + r_e)$
  working space, where $r_e$ is the size of RLBWT of $S_e$.
\end{lemma}

\begin{proof}
  The whole process requires scanning the $\BWT$ of $S_o$ to create a
  run-length compressed encoding of minor subsymbols, adding the $\LF$
  support to (the original) RLBWT of $S_o$, sorting the runs in RLBWT
  of $S_o$ by the minor subsymbol, and executing $r_o$ $\LF$-queries
  on the $\BWT$ of $S_o$, which altogether takes $\bigO(r_o \log
  r_o)$. All other operations take time proportional to $\bigO(r_o +
  r_e)$. The space never exceeds $\bigO(r_o + r_e)$.
\end{proof}

\subsection{Merging BWTs of \texorpdfstring{$S_e$}{Se} and \texorpdfstring{$S_o$}{So}}

As in the previous section, we assume $S=T_i$ for some
$i\in[0..k{-}1]$ and that we are given the RLBWT of $S_o=T_{i+1}$ of
size $r_o=r_{i+1}$ and RLBWT of $S_e$ of size $r_e$. We will show how
to use these to efficiently compute the RLBWT of $S$ in
$\bigO(|S|/\log |S| + (r_o + r_e)\,{\rm polylog}\,|S|)$ time and
space.

We start by observing that to obtain $\BWT$ of $S$ it suffices to
merge the $\BWT$ of $S_e$ and $\BWT$ of $S_o$ and discard all major
subsymbols in the resulting sequence. The algorithm of Hon et
al.~\cite{HonSS03} achieves this by performing the backward search. This
requires $\Omega(|S|)$ time and hence is too expensive in our case.

Instead, we employ the following observation. Suppose we have already
computed the first $t$ runs of the $\BWT$ of $S$ and let the next
unmerged character in the $\BWT$ of $S_o$ be a part of a run of symbol
$c_o$. Let $c_e$ be the analogous symbol from the $\BWT$ of
$S_e$. Further, let $c_e'$ (resp. $c_o'$) be the minor subsymbol of
$c_e$ (resp. $c_o$).  If $c_o' = c_e'$ then either all symbols in the
current run in the $\BWT$ of $S_o$ (restricted to minor subsymbols) or
all symbols in the current run in the (also restricted) $\BWT$ of
$S_e$ will belong to the next run in the $\BWT$ of $S$. Assuming we
can determine the order between any two arbitrary suffixes of $S_o$
and $S_e$ given their ranks in the respective $\BWT$s, we could
consider both cases and in each perform a binary search to find the
exact length of $(t+1)$-th run in the $\BWT$ of $S$. We first locate
the end of the run of $c_o'$ (resp. $c_e'$) in the BWT of $S_o$
(resp. $S_e$) restricted to minor subsymbols; this can be done after
preprocessing input BWTs without increasing the time/space of the
merging.  We then find the largest suffix of $S_e$ (resp. $S_o$) not
greater than the suffix at the end of the run in the BWT of $S_o$.
Importantly, the time to compute the next run in the BWT of $S$ does
not depend on the number of times the suffixes in that run alternate
between $S_o$ and $S_e$. The case $c_e' \neq c_o'$ is handled
similarly, except we do not need to locate the end of each run. The
key property of this algorithm is that the number of pattern searches
is $\bigO(r_i \log |S|)$.

Thus, the merging problem can be reduced to the problem of efficient
comparison of suffixes of $S_e$ and $S_o$. To achieve that we augment
both RLBWTs of $S_e$ and $S_e$ with the suffix-rank support data
structure from Section~\ref{sec:suffix-rank-support}. This will allow
us to determine, given a rank of any suffix of $S_o$, the number of
smaller suffixes of $S_e$ and vice-versa, thus eliminating even the
need for binary search. Our aim is to achieve $\bigO(|S|/\log |S|)$
space and construction time assuming small $r$ values, thus we apply
Theorem~\ref{thm:rank-support} with $\tau=\log^2|S|$.

\begin{lemma}
  \label{lm:merging}
  Given RLBWT of size $r_e$ for $S_e$ and RLBWT of size $r_o$ for $S_o
  = T_{i+1}$ we can compute the RLBWT of $S = T_i$ in
  $\bigO((r_o+r_e)\log^5|S|+|S|/\log |S|+r_i \log^3|S|)$ time and
  $\bigO(|S|/\log^2|S| + (r_o + r_e)\log^4|S| + r_i)$ working space.
\end{lemma}

\begin{proof}
  Constructing the suffix-rank support for $S_o$ and $S_e$ with
  $\tau=\log^2|S|$ takes $\bigO((r_o+r_e)\log^5|S| + |S|/\log |S|)$
  time and $\bigO((r_o+r_e)\log^4|S|+|S|/\log^2|S|)$ working
  space. The resulting data structures occupy $\bigO(|S|/\log^2|S| +
  r_e + r_o)$ space and answer suffix-rank queries in
  $\bigO(\log^3|S|)$ time. To compute the RLBWT of $S$ we perform
  $2r_i$ suffix-rank queries for a total of $\bigO(r_i \log^3|S|)$
  time.
\end{proof}

\subsection{Putting Things Together}

To bound the size of RLBWTs in intermediate rounds, consider the
$i$-th round where for $d=2^i$ we group each $d$ symbols of $\T$ to
obtain the string $S=\T_i$ of length $|T|/d$ and let $r_i$ be the
number of runs in the BWT of $S$. Recall now the construction of
generalized BWT-runs from Section~\ref{sec:suffix-rank-support} and
observe that the symbols of $\T$ comprising each supersymbol $S[j]$
are the same as the substring corresponding to $d$-run containing
suffix $\T[jd+1..n]$ in the $\BWT$ of $\T$. It is easy to see that the
corresponding suffixes of $\T$ are in the same lexicographic order as
the suffixes of $S$. Thus, $r_i$ is bounded by the number of $d$-runs
in the $\BWT$ of $\T$, which by Section~\ref{sec:suffix-rank-support}
is bounded by $rd$. Hence, the size of the output RLBWT of the $i$-th
round does not exceed $r2^i=\bigO(r \log n)$. The analogous analysis
shows that the size of RLBWT of $S_e$ has the same upper bound
$r2^{i+1}$ as $S_o=T_{i+1}$.

\begin{theorem}
  \label{thm:bwt}
  Given string $\T[1..n]$ over alphabet $[1..\sigma]$ of size
  $\sigma\leq n$ encoded in $\bigO(n/\log_{\sigma}n)$ words, the
  $\BWT$ of $\T$ can be computed in $\bigO(n /\log_{\sigma}n + r\log^7
  n)$ time and $\bigO(n/\log_{\sigma}n + r\log^5 n)$ working space,
  where $r$ is the number of runs in the $\BWT$ of $\T$.
\end{theorem}

\begin{proof}
  The $k$-th round of the algorithm takes $\bigO(n/\log_{\sigma}n)$
  time working space and produces a $\BWT$ taking
  $\bigO(n/\log_{\sigma}n)$ words of space. Consider the $i$-th round
  of the algorithm for $i<k$ and let $S=T_{i}$, and $r_e$ and $r_o$
  denote the sizes of RLBWT of $S_e$ and $S_o$ respectively. By the
  above discussion, we have $r_o,r_e=\bigO(r \log n)$. Thus, by
  Lemma~\ref{lm:inducing} and Lemma~\ref{lm:merging} the $i$-th round
  takes $\bigO(n_i/\log n_i + r\log^6 n_i)=\bigO(n/(2^i\log n)+r
  \log^6 n)$ time and the working space does not exceed
  $\bigO(n/\log^2 n + r \log^5 n)$ words, where
  $n_i\,{=}\,|T_i|\,{=}\,n/2^i$, and we used the fact that for $i<k$,
  $\log n_i=\Theta(\log n)$. Hence over all rounds we spend
  $\bigO(n/\log_{\sigma}n + r\log^7 n)$ time and never use more than
  $\bigO(n/\log_{\sigma}n + r\log^5 n)$ space. Finally, it is easy to
  convert RLBWT into the packed encoding in
  $\bigO(n/\log_{\sigma}n+r\log n)$ time.
\end{proof}

Thus, we obtained a time- and space-optimal construction algorithm for
BWT under the assumption $n/r=\Omega({\rm polylog}\,n)$ on the
repetitiveness of the input.

\section{Construction of PLCP}
\label{sec:plcp}

In this section we show that given the run-length compressed
representation of $\BWT$ of $\T$, it is possible to compute the
$\SPLCP$ bitvector in $\bigO(n/\log n + r \log^{11} n)$ time and
$\bigO(n/\log n + r\log^{10} n)$ working space..

The key observation used to construct the PLCP values is that it
suffices to only compute the irreducible LCP values. Then, by
Lemma~\ref{lm:reducible}, all other values can be quickly
deduced. This significantly simplifies the problem because it is known
(Lemma~\ref{lm:irreducible}) that the sum of irreducible LCP values is
bounded by $\bigO(n \log n)$.

The main idea of the construction is to compute (as in
Theorem~\ref{thm:rank-support}) names of $\tau$-runs for $\tau=\log^5
n$. This will allow us to compare $\tau$ symbols at a time and thus
quickly compute a lower bound for large irreducible LCP values. Before
we can use this, we need to augment the BWT with the support for
$\SA$/$\ISA$ queries.

\subsection{Computing \texorpdfstring{$\SA$/$\ISA$}{SA/ISA} Support}
\label{sec:sa-support}

Suppose that we are given a run-length compressed $\BWT$ of $\T[1..n]$
taking $\bigO(r)$ space. Let $\tau\geq 1$ be an integer. Assume for
simplicity that $n$ is a multiple of $\tau$. We start by computing the
sorted list of starting positions of all $\tau$-runs similarly, as in
Theorem~\ref{thm:rank-support}. This requires augmenting the RLBWT
with the $\LF$/$\PSI$ support first and in total takes $\bigO(\tau r
\log (\tau r))$ time and $\bigO(\tau r)$ working space. We then
compute and store, for the first position of each $\tau$-run $[b..e]$,
the value of $\LF^{\tau}[b]$. This will allow us to efficiently
compute $\LF^{\tau}[i]$ for any $i\in[1..n]$.

We then locate the occurrence $i_0$ of the symbol $\$$ in $\BWT$ and
perform $n/\tau$ iterations of $\LF^{\tau}$ on $i_0$.  By definition
of $\LF$, the position $i$ visited after $j$ iterations of
$\LF^{\tau}$ is equal to $\ISA[n-j\tau]$, i.e., $\SA[i]=n-j\tau$. For
any such $i$ we save the pair $(i,n-j\tau)$ into a list. When we
finish the traversal we sort the list by the first component (assume
this list is called $L_{\SA}$). We then create the copy of the list
(call it $L_{\ISA}$) and sort it by the second component.  Creating
the lists takes $\bigO\left((n/\tau)(\log (r\tau) + \log
(n/\tau))\right)$ time and they occupy $\bigO(n/\tau)$ space. After
the lists are computed we discard $\LF^{\tau}$ samples associated with
all runs.  Having these lists allows us to efficiently query
$\SA$/$\ISA$ as follows.

To compute $\ISA[j]$ we find in $\bigO(1)$ time (since we can store
$L_{\ISA}$ in an array) the pair $(p,j')$ in $L_{\ISA}$ such that
$j'=\lceil j/\tau \rceil \tau$. We then perform $j'-j<\tau$ steps of
$\LF$ on position $p$. The total query time is thus $\bigO(\tau \log
r)$.

To compute $\SA[i]$ we perform $\tau$ steps of $\LF$ (each taking
$\bigO(\log r)$ time) on position $i$. Due to the way we sampled
$\SA$/$\ISA$ values, one of the visited positions has to be the first
component in the $L_{\SA}$ list. For each position, we can check this
in $\bigO(\log(n/\tau))$ time. Suppose we found a pair after
$\Delta<\tau$ steps, i.e., a pair $(\LF^{\Delta}[i],p)$ is in
$L_{\SA}$. This implies $\SA[\LF^{\Delta}[i]]=p$, i.e.,
$\SA[i]=p+\Delta$.  The query time is $\bigO\left(\tau (\log r + \log
(n/\tau))\right)$.

\begin{theorem}
  \label{thm:sa-support}
  Given RLBWT of size $r$ for text $T[1..n]$, we can, for any integer
  $\tau\geq 1$, build a data structure taking $\bigO(r+n/\tau)$ space
  that, for any $i\in[1..n]$, can answer $\SA[i]$ query in $\bigO(\tau
  (\log r + \log(n/\tau)))$ time and $\ISA[i]$ query in
  $\bigO\left(\tau \log r \right)$ time. The construction takes
  $\bigO\left((n/\tau) (\log (r\tau) + \log
  (n/\tau))+\tau^2r\log(r\tau)\right)$ time and $\bigO(n/\tau +
  r\tau)$ working space.
\end{theorem}

\subsection{Computing Irreducible LCP Values}
\label{sec:irreducible}

We start by augmenting the RLBWT with the $\SA$/$\ISA$ support as
explained in the previous section using $\tau_1=\log^2 n$. The
resulting data structure answers $\SA$/$\ISA$ queries in $\bigO(\log^3
n)$ time. We then compute $\tau_2$-runs and their names using the
technique introduced in Theorem~\ref{thm:rank-support} for
$\tau_2=\log^5 n$.

Given any $j_1,j_2\in[1..n]$ we can check whether it holds
$\T[j_1..j_1+\tau_2-1]=\T[j_2..j_2+\tau_2-1]$ using the above names as
follows.  Compute $i_1=\ISA[j_1+\tau_2]$ and $i_2=\ISA[j_2+\tau_2]$
using the $\ISA$ support. Then compare the names of
$\tau_2$-substrings preceding these two suffixes. Thus, comparing two
arbitrary substrings of $\T$ of length $\tau_2$, given their text
positions, takes $\bigO(\log^3 n)$ time.

The above toolbox allows computing all irreducible LCP values as
follows.  For any $i\in[1..n]$ such that $\LCP[i]$ is irreducible
(such $i$ can be recognized by checking if $\BWT[i-1]$ belongs to a
BWT-run different than $\BWT[i]$) we compute $j_1=\SA[i-1]$ and
$j_2=\SA[i]$. We then have $\LCP[i]=\lcp(\T[j_1..n],\T[j_2..n])$. We
start by computing the lower-bound for $\LCP[i]$ using the names of
$\tau_2$-substrings. Since the sum of irreducible LCP values is
bounded by $\bigO(n \log n)$, over all irreducible LCP values this
will take $\bigO(r\log^3 n + \log^3 n \cdot (n \log n) /
\tau_2)=\bigO(r\log^3 n + n/\log n)$ time.  Finishing the computation
of each $\LCP$ value requires at most $\tau_2$ symbol
comparisons. This can be done by following $\PSI$ for both pointers as
long as the preceding symbols (found in the $\BWT$) are equal.  Over
all irreducible LCP values, finishing the computation takes $\bigO(r
\tau_2 \log n)=\bigO(r \log^6 n)$ time.

\begin{theorem}
  \label{thm:plcp}
  Given RLBWT of size $r$ for $\T[1..n]$, the $\SPLCP$ bitvector (or
  the list storing irreducible LCP values in text order) can be
  computed in $\bigO(n/\log n + r \log^{11} n)$ time and $\bigO(n/\log
  n + r\log^{10}n)$ working space.
\end{theorem}

\begin{proof}
  Adding the $\SA$/$\ISA$ support using $\tau_1=\log^2 n$ takes
  $\bigO\left(n/\log n + r\log^5 n\right)$ time and $\bigO(n/\log^2 n
  + r\log^2 n)$ working space (Theorem~\ref{thm:sa-support}).  The
  resulting structure needs $\bigO(r + n/\log^2 n)$ space and answers
  $\SA$/$\ISA$ queries in $\bigO(\log^3 n)$ time.

  Computing the names takes $\bigO(\tau_2^2r \log (\tau_2 r))=\bigO(r
  \log^{11} n)$ time and $\bigO(\tau_2^2 r)=\bigO(r \log^{10} n)$
  working space (see the proof of Theorem~\ref{thm:rank-support}). The
  names need $\bigO(\tau_2 r)=\bigO(r \log^5 n)$ space.  Then, by the
  above discussion, computing all irreducible LCP values takes
  $\bigO(n/\log n + r \log^6 n)$ time.
\end{proof}

By combining with Theorem~\ref{thm:bwt} we obtain the following
result.

\begin{theorem}
  \label{cor:plcp}
  Given string $\T[1..n]$ over alphabet $[1..\sigma]$ of size
  $\sigma\,{\leq}\, n$ encoded in $\bigO(n/\log_{\sigma}n)$ words, the
  $\SPLCP$ bitvector (or the list storing irreducible LCP values in
  text order) can be computed in $\bigO(n/\log_{\sigma} n + r
  \log^{11} n)$ time and $\bigO(n/\log_{\sigma} n + r\log^{10}n)$
  working space, where $r$ is the number of runs in the $\BWT$ of
  $\T$.
\end{theorem}

\section{Construction of RLCSA}
\label{sec:rlcsa}

In this section, we show how to use the techniques presented in this paper to quickly
build the run-length compressed suffix array (RLCSA) recently proposed
by Gagie et al.~\cite{GagieNP17}. They observed that if $\BWT$ of $T$
has $r$ runs then the arrays $\SA/\ISA$ and $\LCP$ have a
bidirectional parse of size $\bigO(r)$ after being differentially
encoded.  They use a locally-consistent parsing~\cite{BatuES06,Jez15}
to grammar-compress these arrays and describe the necessary
augmentations to achieve fast decoding of the original values.  This
allowed them to obtain a $\bigO(r\,{\rm polylog}\,n)$-space structure
that can answer $\SA$/$\ISA$ and $\LCP$ queries in $\bigO(\log n)$
time.

The structure described below is slightly different than the original
index proposed by Gagie et al.~\cite{GagieNP17}. Rather than
compressing the differentially-encoded suffix array, we directly
exploit the structure of the array. It can be thought of as a
multi-ary block tree~\cite{DCC2015} modified to work with arrays
indexed in ``lex-order'' instead of the original ``text-order''.  Our
data structure matches the space and query time of~\cite{GagieNP17},
but we additionally show how to achieve a trade-off between space
and query time. In particular, we achieve $\bigO(\log n/\log \log n)$
query time in $\bigO(r\,{\rm polylog}\, n)$ space.

\subsection{Data Structure}
\label{sec:small-space-sa-support}

Suppose we are given RLBWT of size $r$ for text $T[1..n]$.  The data
structure is parametrized by an integer parameter $\tau>1$.  For
simplicity, we assume that $r$ divides $n$ and that $n/r$ is a power of
$\tau$. The data structure is organized into $\log_{\tau}(n/r)$
\emph{levels}.  The main idea is, for every level, to store $2\tau$
pointers for each BWT-run boundary. The purpose of pointers is to
reduce the $\SA$ query near the associated run boundary into $\SA$
query at a position that is closer (by at least a factor of $\tau$) to
some (usually different) run boundary. Level controls the allowed
proximity of the query.  At the last level, the $\SA$ value at each
run boundary is stored explicitly.

More precisely, for $1 \leq k \leq \log_{\tau}(n/r)$, let
$b_k=n/(r\tau^k)$ and let $\BWT[b..e]$ be one of the runs in the
$\BWT$.  Consider $2\tau$ non-overlapping consecutive blocks of size
$b_k$ evenly spread around position $b$, i.e.,
$\BWT[b+ib_k..b+(i+1)b_k-1]$, $i=-\tau,\ldots,\tau-1$.  For each block
$\BWT[s..t]$ we store the smallest $d$ (called \emph{LF-distance})
such that there exists at least one $i\in[s..t]$ such that $\LF^d[i]$
is the beginning of the run in the $\BWT$ of $\T$ (note that it is
possible that $d=0$). With each block we also store the value
$\LF^d[s]$ (called \emph{LF-shortcut}), both as an absolute value in
$[1..n]$ and as a pointer to the BWT-run that contains it. Due to
the simple generalization of Lemma~\ref{lm:lf-in-run}, this allows us to
compute $\LF^d[i]$ for \emph{any} $i\in[s..t]$.  At each level, we
store $2\tau$ integers for each of $r$ BWT runs thus in total we store
$\bigO(r\tau\log_{\tau}(n/r))$ words.

To access $\SA[i]$ we proceed as follows. Assume first that $i$ is not
more than $n/r$ positions from the closest run boundary.  We first
find the $\BWT$ run that contains $i$. We then follow the
$\LF$-shortcuts starting at level 1 down to the last level. After
every step, the distance to the closest run boundary is reduced by a
factor $\tau$. Thus, after $\log_{\tau}(n/r)$ steps the current
position is equal to boundary $b$ of some run $\BWT[b..e]$. Let
$d_{\rm sum}$ denote the total lengths of $\LF$-distances of the used
shortcuts. Since $\SA[b]$ is stored we can now answer the query as
$\SA[i]=\SA[b]+d_{\rm sum}$. To handle positions further than $n/r$
from the nearest run boundary, we add a lookup table $LT[1..r]$ such
that $LT[i]$ stores the $\LF$-shortcut and $\LF$-distance for block
$\BWT[(i-1)(n/r)+1..i(n/r)]$. The query time is
$\bigO(\log_{\tau}(n/r))$, since blocks in the same level have the
same length and hence at each level we spend $\bigO(1)$ time to find
the pointer to the next level.  Note that the lookup table eliminates
the initial search of run containing $i$.

The above data structure can be generalized to extract segments of
$\SA[p..p+\ell-1]$, for any $p$ and $\ell$, faster than $\ell$ single
$\SA$-accesses, that would cost $\bigO(\ell \log_{\tau}(n/r))$. The
main modification is that at level $k$ we instead consider $4\tau-1$
blocks of size $b_k$, evenly spread around position $b$, each
overlapping the next by exactly $b_k/2$ symbols, i.e.,
$\BWT[b+ib_k/2..b+(i+2)b_k/2-1]$, $i=-2\tau,\ldots,2(\tau-1)$. This
guarantees that any segment-access to $\SA$ of length at most $b_k/2$
at level $k$ can be transformed into the segment-access at level
$k+1$. We also truncate the data structure at level $k$ where $k$ is
the smallest integer with $b_{k} < \log_{\tau}(n/r)$. At that level we
store a segment of $2\log_{\tau}(n/r)$ $\SA$ values around each BWT
run. These values take $\bigO(r \log_{\tau}(n/r))$ space, and hence
the two modifications do not increase the space needed by the data
structure. This way we can extract $\SA[p..p+\alpha-1]$, where
$\alpha=\log_{\tau}(n/r)$ in $\bigO(\alpha)$ time, and consequently a
segment $\SA[p..p+\ell-1]$ in
$\bigO((\ell/\alpha+1)\alpha)=\bigO(\ell+\log_{\tau}(n/r))$ time.

\begin{theorem}
  \label{thm:sa}
  Assume that $\BWT$ of $\T[1..n]$ consist of $r$ runs.  For any
  integer $\tau{>}1$, there exists a data structure of size $\bigO(r
  \tau \log_{\tau}(n/r))$ that, for any $p\in[1..n]$ and $\ell\geq 1$,
  can compute $\SA[p..p+\ell-1]$ in $\bigO(\ell+\log_{\tau}(n/r))$
  time.
\end{theorem}

For $\tau=2$ the above data structure matches the space and query time
of~\cite{GagieNP17}.  For $\tau=\log^{\epsilon}n$, where $\epsilon>0$
is an arbitrary constant it achieves $\bigO(r \log^{\epsilon} n
\log(n/r))$ space and $\bigO(\log n / \log \log n)$ query time.
Finally, for $\tau=(n/r)^{\epsilon}$ it achieves
$\bigO(r^{1-\epsilon}n^{\epsilon})$ space and $\bigO(1)$ time
query. In particular, if $r=o(n)$ the data structure takes $o(n)$
space and is able to access (any segment of) $\SA$ in optimal time.

\subsection{Construction Algorithm}

Assume we are given the run-length compressed $\BWT$ of $\T[1..n]$ of
size $r$. Consider any block $\BWT[s..t]$. Let $d$ be the
corresponding $\LF$-distance and let $\LF^d[i]=b$ for some
$i\in[s..t]$ be the beginning of a BWT-run $[b..e]$. We observe that
this implies $\LCP[b]$ is irreducible and $\LCP[b] \geq d$.

We start by augmenting the RLBWT with the $\SA$/$\ISA$ support from
Section~\ref{sec:sa-support} using $\tau_1=\log^2 n$. This, by
Theorem~\ref{thm:sa-support}, takes $\bigO\left(n/\log n + r\log^5
n\right)$ time and $\bigO(n/\log^2 n + r\log^2 n)$ working space.  The
resulting structure needs $\bigO(r + n/\log^2 n)$ space and allows
answering $\SA$/$\ISA$ queries in $\bigO(\log^3 n)$ time.

Consider now the sorted sequence $Q$ containing every position $j$ in
$T$ such that $\PLCP[j]$ is irreducible. Such list can be obtained by
computing value $\SA[b]$ for every BWT run $[b..e]$ and sorting the
resulting values. Computing the list $Q$ takes $\bigO(r \log^3 n)$
time and $\bigO(r)$ working space. The list itself is stored in plain
form using $\bigO(r$) space.  Next, for any irreducible value
$\PLCP[j]$ we compute, for any $t=1,\ldots,\lfloor \ell'/\tau_2
\rfloor$ a pair containing $\ISA[j+t\tau_2]$ (as key) and $t\tau_2$
(as value), where $\tau_2=\log^4 n$, and $\ell'$ is the distance
between $j$ and its successor in $Q$.  Since the sum of $\ell'$ values
is $\bigO(n)$, computing all pairs takes $\bigO(\log^3 n \cdot (r + n
/ \tau_2))=\bigO(n/\log n+r \log^3 n)$ time and
$\bigO(n/\tau_2)=\bigO(n/\log^4 n)$ working space. The resulting pairs
need $\bigO(n / \log^4 n)$ space.

We then sort all the computed pairs by the keys and build a static RMQ
data structure over the associated values. This can be done in
$\bigO\left(n/\tau_2\right)=\bigO(n/\log^4 n)$ time and space so that
an RMQ query takes $\bigO(\log n)$ time (using static balanced BST).

Having the above samples augmented with the RMQ allows us to compute
$\LF$-shortcuts as follows. Let $\BWT[s..t]$ be one of the blocks. We
perform $\tau_2$ $\LF$-steps on position $s$. In step $\Delta$ we
first check in $\bigO(\log r)$ time whether the block
$[\LF^{\Delta}[s]..\LF^{\Delta}[s]+(t-s)]$ contains a boundary of a
BWT-run. If yes, then we found the $\LF$-distance and terminate the
procedure. Otherwise, in $\bigO(\log n)$ we compute the minimal value
$d_{\min}$ and its position for the block
$[\LF^{\Delta}[s]..\LF^{\Delta}[s]+(t-s)]$ using the RMQ structure (if
the block is empty we skip this step). We call $d_{\min}+\Delta$ the
\emph{candidate value}.  From the way we computed the pairs, the
minimum candidate value is equal to the $\LF$-distance of
$\BWT[s..t]$.  It is easy to extend this procedure to also return the
$\LF$-shortcut.

Thus, the $\LF$-shortcut for any block can be computed in
$\bigO(\tau_2 \log n)=\bigO(\log^5 n)$ time. Over all blocks (and
including the shortcuts for the lookup table $LT[1..r]$) this takes
$\bigO(r \tau \log_{\tau}(n/r) \log^5 n)=\bigO(r \tau \log^6 n)$ time.
Finally, computing segments of $\SA$ values at the last level (after
truncating the tree) takes $\bigO(r \log_{\tau}(n/r) \log^3 n)$ time.

\begin{theorem}
  \label{thm:sa-construct-from-RLBWT}
  Given RLBWT of size $r$ for text $\T[1..n]$ we can build the data
  structure from Theorem~\ref{thm:sa} in $\bigO(n/\log n + r\tau\log^6
  n)$ time and $\bigO(n/\log^2 n + r(\tau\log_{\tau}(n/r)+ \log^2 n))$
  working space.
\end{theorem}

By combining with Theorem~\ref{thm:bwt} we obtain the following
theorem.

\begin{theorem}
  \label{thm:sa-construct-from-text}
  Given string $\T[1..n]$ over alphabet $[1..\sigma]$ of size
  $\sigma\,{\leq}\,n$ encoded in $\bigO(n/\log_{\sigma}n)$ words we
  can build the data structure from Theorem~\ref{thm:sa} in
  $\bigO(n/\log_{\sigma} n + r(\tau\log^6 n+\log^{7} n))$ time and
  $\bigO(n/\log_{\sigma} n + r(\tau\log_{\tau}(n/r)+\log^{5}n))$
  working space, where $r$ is the number of runs in the $\BWT$ of
  $\T$.
\end{theorem}

\section{Construction of LZ77 Parsing}
\label{sec:lz77}

In this section, we show how to use the techniques introduced in
previous sections to obtain a fast and space-efficient LZ77
factorization algorithm for highly repetitive strings.

\subsection{Definitions}

The LZ77 factorization~\cite{LZ77} uses the notion of the {\em longest
  previous factor} (LPF).  The LPF at position $i$ (denoted $\LPF[i]$)
in $\T$ is a pair $(p_i,\ell_i)$ such that, $p_i < i$,
$\T[p_i..p_i+\ell_i-1] = \T[i..i+\ell_i-1]$ and $\ell_i>0$ is
maximized.  In other words, $\T[i..i+\ell_i-1]$ is the longest prefix
of $\T[i..n]$ which also occurs at some position $p_i < i$ in $\T$. If
$\T[i]$ is the leftmost occurrence of a symbol in $\T$ then such a
pair does not exist. In this case we define $p_i = \T[i]$ and $\ell_i
= 0$. Note that there may be more than one potential $p_i$, and we do
not care which one is used.

The LZ77 factorization (or LZ77 parsing) of a string $\T$ is then just
a greedy, left-to-right parsing of $\T$ into longest previous
factors. More precisely, if the $j^{\mbox{{\scriptsize th}}}$ LZ
factor (or {\em phrase}) in the parsing is to start at position $i$,
then we output $(p_i,\ell_i)$ (to represent the $j^{\mbox{{\scriptsize
      th}}}$ phrase), and then the $(j+1)^{\mbox{{\scriptsize th}}}$
phrase starts at position $i+\ell_i$, unless $\ell_i = 0$, in which
case the next phrase starts at position $i+1$. For the example string
$\T = zzzzzipzip$, the LZ77 factorization produces:
$$(z,0),(1,4),(i,0),(p,0),(5,3).$$ We denote the number of phrases in
the LZ77 parsing of $\T$ by $z$. The following theorem shows that LZ77
parsing can be encoded in $\bigO(n \log \sigma)$ bits.

\begin{theorem}[e.g.~\cite{phdjuha}]
  \label{thm:lz77-size}
  The number of phrases $z$ in the LZ77 parsing of a text of $n$
  symbols over an alphabet of size $\sigma$ is
  $\bigO(n/\log_{\sigma}n)$.
\end{theorem}

The LPF pairs can be computed using \emph{next and previous smaller
  values} (NSV/PSV) defined as
\begin{align*}
  \NSVL[i] &= \min \{ j\in [i+1..n] \mid \SA[j] < \SA[i]\},\\
  \PSVL[i] &= \max \{ j\in [1..i-1] \mid \SA[j] < \SA[i]\} .
\end{align*}
If the set on the right hand side is empty, we set the value to $0$.
We further define
\begin{align*}
  \NSVT[i] &= \SA[\NSVL[\ISA[i]]],\\
  \PSVT[i] &= \SA[\PSVL[\ISA[i]]].
\end{align*}
If $\NSVL[\ISA[i]]=0$ ($\PSVL[\ISA[i]]=0$) we set $\NSVT[i]=0$
($\PSVT[i]=0$).

The usefulness of the NSV/PSV values is summarized by the following
lemma.

\begin{lemma}[\cite{ci2008}]
  \label{lm:psv-nsv}
  For $i\in[1..n]$, let $i_{nsv}=\NSVT[i]$, $i_{psv}=\PSVT[i]$,
  $\ell_{nsv} = \lcp(i,i_{nsv})$ and $\ell_{psv} =
  \lcp(i,i_{psv})$. Then
  \[
    \LPF[i] = \left\{
      \begin{array}{ll}
        (i_{nsv},\ell_{nsv}) & \text{ if\enspace} \ell_{nsv} > \ell_{psv}, \\
        (i_{psv},\ell_{psv}) & \text{ if\enspace} \ell_{psv} =
        \max(\ell_{nsv},\ell_{psv}) > 0, \\
        (\T[i],0) & \text{ if\enspace} \ell_{nsv} = \ell_{psv} = 0.
      \end{array}
    \right.
  \]
\end{lemma}

\subsection{Algorithm Overview}

The general approach of our algorithm follows the lazy LZ77
factorization algorithms of~\cite{kkp-jea}. Namely, we opt out from
computing all $\LPF$ values and instead compute $\LPF[j]$ only when
there is an LZ factor starting at position $j$.

Suppose we have already computed the parsing of $\T[1..j-1]$. To
compute the factor starting at position $j$ we first query
$i=\ISA[j]$. We then compute (using a small-space data structure
introduced next) values $i_{\rm nsv}=\NSVL[i]$ and $i_{\rm
  psv}=\PSVL[i]$. By Lemma~\ref{lm:psv-nsv} it then suffices to
compute the lcp of $\T[j..n]$ and each of the two suffixes starting at
positions $\SA[i_{\rm psv}]$ and $\SA[i_{\rm nsv}]$.

It is easy to see that the total length of computed lcps will be
$\bigO(n)$, since after each step we increase $j$ by the longest of
the two lcps.  To perform the lcp computation efficiently we will
employ the technique from Section~\ref{sec:plcp} which allows
comparing multiple symbols at a time. This will allow us to spend
$\bigO(z\,{\rm polylog}\,n + n/\log n)$ time in the lcp computation.
The problem is thus reduced to being able to quickly answer
$\NSVL/\PSVL$ queries.

\subsection{Computing \texorpdfstring{$\NSV/\PSV$}{NSV/PSV} Support for \texorpdfstring{$\SA$}{SA}}
\label{sec:nsv-psv}

Assume that we are given RLBWT of size $\bigO(r)$ for text $\T[1..n]$.
We will show how to quickly build a small-space data structure that,
given any $i\in[1..n]$ can compute $\NSVL[i]$ or $\PSVL[i]$ in
$\bigO({\rm polylog}\,n)$ time.

We split $\BWT[1..n]$ into blocks of size $\tau=\Theta({\rm
  polylog}\,n)$ and for each $j\in[1..n/\tau]$ we compute the minimum
value in $\SA[(j{-}1)\tau{+}1..j\tau]$ together with its position. We
then build a balanced binary tree over the array of minimas and
augment each internal node with the minimum value in its subtree. This
allows, for any $j\in[1..n/\tau]$, and any value $x$, to find the
maximal (resp. minimal) $j'<j$ (resp. $j'>j$) such that
$\SA[(j'{-}1)\tau{+}1..j'\tau]$ contains a value smaller than $x$.  At
query time we first scan the $\SA$ positions preceding or following
the query position $i\in[1..n]$ inside the block containing $i$.  If
there is no value smaller than $\SA[i]$, we use the RMQ to find the
closest block with a value smaller than $\SA[i]$. To finish the query it
then suffices to scan the $\SA$ values inside that block.  It takes
$\bigO(\log^3 n)$ time to compute $\SA$ value
(Theorem~\ref{thm:sa-support}), hence answering a single
$\NSVL$/$\PSVL$ query will take $\bigO(\tau \log^3 n)$.

To compute the minimum for each of the size-$\tau$ blocks of $\SA$ we
observe that, up to a shift by a constant, there is only $r\tau$
different blocks.  More specifically, consider a block
$\SA[(j{-}1)\tau{+}1..j\tau]$. Let $k$ be the smallest integer such
that for some $t\in[(j{-}1)\tau+1..j\tau]$, $\LF^k[t]$ is the
beginning of a run in $\BWT$. It is easy to see that, due to
Lemma~\ref{lm:lf-in-run},
$\SA[(j{-}1)\tau{+}1..j\tau]=k+\SA[\LF^k[j\tau]{-}\tau{+}1..\LF^k[j\tau]]$,
in particular, the equality holds for the minimum element. Thus, it
suffices to precompute the minimum value and its position for each of
the $r\tau$ size-$\tau$ blocks intersecting a boundary of a
$\BWT$-run. This takes $\bigO(r\tau \log^3 n)$ time and $\bigO(r\tau)$
working space. The resulting values need $\bigO(r\tau)$ space.

It thus remains to compute the ``$\LF$-distance'' for each of the
$n/\tau$ blocks of $\SA$, i.e., the smallest $k$ such that for at
least one position $t$ inside the block, $\LF^k[t]$ is the beginning
of a $\BWT$-run. To achieve this we utilize the technique used in
Section~\ref{sec:rlcsa}. There we presented a data structure of size
$\bigO(r + n/\log^2 n + n/\tau_2)$ that can be built in $\bigO(n/\log
n + r\log^5 n + (n \log^3 n) / \tau_2)$ time and $\bigO(n/\log^2 n +
r\log^2 n + n/\tau_2)$ working space, and is able to compute the
$\LF$-shortcut for any block $[s..t]$ in $\SA$ in $\bigO(\tau_2 \log
n)$ time.

\begin{theorem}
  \label{thm:nsv-psv}
  Given RLBWT of size $r$ for text $\T[1..n]$, we can build a data
  structure of size $\bigO(r+n/\log^2 n)$ that can answer
  $\PSVL$/$\NSVL$ queries in $\bigO(\log^{9} n)$ time.  The data
  structure can be built in $\bigO(n/\log n+r\log^{9}n)$ time and
  $\bigO(n/\log^2 n + r\log^6 n)$ working space.
\end{theorem}
\begin{proof}
  We start by augmenting the RLBWT with $\SA$/$\ISA$ support. This
  takes (Theorem~\ref{thm:sa-support}) $\bigO(n/\log n+r\log^5 n)$
  time and $\bigO(n/\log^2 n + r\log^2n)$ working space. The resulting
  data structure takes $\bigO(r+n/\log^2 n)$ space and answers
  $\SA$/$\ISA$ queries in $\bigO(\log^3 n)$ time.

  To achieve the $\bigO(n/\log n)$ term in the construction time for
  the structure from Section~\ref{sec:rlcsa} we set $\tau_2=\log^4
  n$. Then, computing the $\LF$-shortcut for any block in $\SA$ takes
  $\bigO(\log^5 n)$ time. Since we have $n/\tau$ blocks to query, we
  set $\tau=\log^6 n$ to obtain $\bigO(n/\log n)$ total query
  time. Answering a single $\NSVL$/$\PSVL$ query then takes
  $\bigO(\tau \log^3 n)=\bigO(\log^{9} n)$.

  The RMQ data structure built on top of the minimas of the blocks of
  $\SA$ takes $\bigO(n/\tau)=\bigO(n/\log^6 n)$ space, hence the space
  of the final data structure is dominated by $\SA$/$\ISA$ support
  taking $\bigO(r+n/\log^2 n)$ words.

  The construction time is split between precomputing the minimas in
  each of the $r\tau$ blocks crossing boundaries of $\BWT$-runs in
  $\bigO(r\tau \log^3 n)=\bigO(r \log^{9}n)$ time, and other steps
  introducing term $\bigO(n/\log n)$.

  The working space is maximized when building the $\SA$/$\ISA$
  support and during the precomputation of minimas in each of the
  $r\tau$ blocks, for a total of $\bigO(n/\log^2 n + r\log^6 n)$.
\end{proof}

\subsection{Algorithm Summary}

\begin{theorem}
  \label{thm:lz77}
  Given RLBWT of size $r$ of $\T[1..n]$, the LZ77 factorization of
  $\T$ can be computed in $\bigO(n/\log n + r \log^{9}n + z\log^{9}n)$
  time and $\bigO(n/\log^2 n + z + r\log^8 n)=\bigO(n/\log_{\sigma}n +
  r\log^8 n)$ working space, where $z$ is the size of the LZ77 parsing
  of $\T$.
\end{theorem}
\begin{proof}
  We start by augmenting the RLBWT with the $\SA$/$\ISA$ support from
  Section~\ref{sec:sa-support} using $\tau_1=\log^2 n$.  This, by
  Theorem~\ref{thm:sa-support}, takes $\bigO\left(n/\log n + r\log^5
  n\right)$ time and $\bigO(n/\log^2 n + r\log^2 n)$ working space.
  The resulting structure needs $\bigO(r + n/\log^2 n)$ space and
  answers $\SA$/$\ISA$ queries in $\bigO(\log^3 n)$ time.

  Next, we initialize the data structure supporting the
  $\PSVL$/$\NSVL$ queries from Section~\ref{sec:nsv-psv}. By
  Theorem~\ref{thm:nsv-psv} the resulting data structure needs
  $\bigO(r+n/\log^2 n)$ space and answers queries in $\bigO(\log^{9}
  n)$ time. The data structure can be built in $\bigO(n/\log
  n+r\log^{9}n)$ time and $\bigO(n/\log^2 n + r\log^6 n)$ working
  space. Over the course of the whole algorithm, we ask $\bigO(z)$
  queries hence in total we spend $\bigO(z\log^{9}n)$ time.

  Lastly, we compute $\tau_3$-runs and their names using the technique
  introduced in Section~\ref{sec:irreducible} for $\tau_3=\log^4 n$.
  This takes $\bigO(\tau_3^2r \log (\tau_3 r))=\bigO(r \log^{9} n)$
  time and $\bigO(\tau_3^2 r)=\bigO(r \log^{8} n)$ working space (see
  the proof of Theorem~\ref{thm:plcp}).  The names need $\bigO(\tau_3
  r)=\bigO(r \log^4 n)$ space.  The names allow, given any $j_1,j_2\in
  [1..n]$, to compute $\ell=\lcp(j_1,j_2)$ in $\bigO\left(\log^3
  n(1+\ell/\tau_3)+\tau_3 \log n\right)=\bigO\left(\log^5 n + \ell /
  \log n\right)$ time. Thus, over the course of the whole algorithm we
  will spend $\bigO(z \log^5 n + n/\log n)$ time computing lcp values.
\end{proof}

By combining with Theorem~\ref{thm:plcp} we obtain the following
result.

\begin{theorem}
  \label{thm:lz77-2}
  Given string $\T[1..n]$ over alphabet $[1..\sigma]$ of size
  $\sigma\leq n$ encoded in $\bigO(n/\log_{\sigma}n)$ words, we can
  compute the LZ77 factorization of $\T$ in
  $\bigO(n/\log_{\sigma}n+r\log^{9}n+z\log^{9}n)$ time and
  $\bigO(n/\log_{\sigma}n+r\log^8 n)$ working space, where $r$ is the
  number of runs in the $\BWT$ of $\T$ and $z$ is the size of the LZ77
  parsing of $\T$.
\end{theorem}

Since $z=\bigO(r \log n)$~\cite{GNPlatin18}, the above algorithm
achieves $\bigO(n/\log_{\sigma}n)$ runtime and working space when
$n/r\in\Omega({\rm polylog}\,n)$.

\section{Construction of Lyndon Factorization}

In this section, we show another application of our techniques. Namely,
we show that we can obtain a fast and space-efficient construction of
Lyndon factorization for highly repetitive strings.

\subsection{Definitions}

A string $S$ is called a \emph{Lyndon word} if $S$ is
lexicographically smaller than all its non-empty proper suffixes.  The
\emph{Lyndon factorization} (also called \emph{Standard
  factorization}) of a string $T$ is its unique (see \cite{CFL58})
factorization $T=f_1^{e_1} \cdots f_m^{e_m}$ such that each $f_i$ is a
Lyndon word, $e_{i} \geq 1$, and $f_{i} \succ f_{i+1}$ for all $1 \leq
i < m$.  We call each $f_i$ a \emph{Lyndon factor} of $T$, and each
$F_i= f_i^{e_i}$ a \emph{Lyndon run} of $T$. The size of the Lyndon
factorization is $m$, the number of distinct Lyndon factors, or
equivalently, the number of Lyndon runs.

Each Lyndon run can be encoded as a triple of integers storing the
boundaries of some occurrence of $f_i$ in $T$ and the exponent $e_i$.
Since, for any string, it holds $m<2z$~\cite{KarkkainenKNPS17} and
$z=\bigO(n/\log_{\sigma}n)$~\cite{phdjuha}, where $z$ is the number of
phrases in the LZ77 parsing, it follows that Lyndon factorization can
be stored in $\bigO(n \log \sigma)$ bits.

\subsection{Algorithm Overview}

Our algorithm utilizes many of the algorithms from the long line of
research on algorithms operating on compressed representations such as
grammars or LZ77 parsing:
\begin{itemize}
  \item Furuya et al.~\cite{CPM2018} have shown that given an SLP
    (i.e., a grammar in Chomsky normal form generating a single
    string) of size $g$ generating string $T$ of length $n$, the
    Lyndon factorization of $T$ can be computed in $\bigO(P(g, n) +
    Q(g, n)\,g \log \log n)$ time and $\bigO(g \log n + S(g, n))$
    space, where $P(g, n)$, $S(g, n)$, $Q(g, n)$ are respectively the
    pre-processing time, space, and query time of a data structure for
    longest common extensions (LCE) queries on SLPs. The LCE query,
    given two positions $i$ and $j$ in the string $T$, returns
    $\lcp(i,j)$, i.e., the length of the longest common prefix of
    suffixes $T[i..n]$ and $T[j..n]$.
  \item On the other hand, Nishimoto et al.~\cite[Thm
    3]{NishimotoIIBT16} have shown how, given an SLP of size $g$
    generating string $T$ of length $n$, to construct an LCE data
    structure in $\bigO(g \log \log g \log n \log^{*}n)=\bigO(g \log^3
    n)$ time and $\bigO(g \log^{*}n + z \log n \log^{*}n)=\bigO(g
    \log^2 n)$ space, where $z$ is the size of LZ77 parsing of
    $T$. The resulting data structure answers a query ${\rm LCE}(i,j)$
    in $\bigO(\log n + \log \ell \log^{*}n)=\bigO(\log^2 n)$ time,
    where $\ell=\lcp(i,j)$. Thus, they achieve $P(g, n)=\bigO(g \log^3
    n)$, $S(g, n)=\bigO(g \log^2 n)$, and $Q(g, n)=\bigO(\log^2
    n)$. More recently, I~\cite[Thm 2]{I17} improved (using different
    techniques) this to $P(g, n)=\bigO(g \log (n/g))$, $S(g,
    n)=\bigO(g+z\log(n/z))$, and $Q(g, n)=\bigO(\log n)$.
  \item Finally, Rytter~\cite[Thm 2]{Rytter03} have shown how, given
    the LZ77 parsing of string $T$ of length $n$, to convert it into
    an SLP of size $g=\bigO(z \log n)$ in $\bigO(z \log n)$ time and
    $\bigO(z \log n)$ working space.
\end{itemize}

The above pipeline leads to a fast and space-efficient algorithm for
Lyndon factorization, assuming the compressed representation (such as
SLP or LZ77) of text is given \emph{a priori}. It still, however,
needs $\Omega(n)$ time if we take into account the time to compute
LZ77 or a small grammar using the previously fastest known algorithms.
Section~\ref{sec:lz77} completes this line of
research by providing fast and space-efficient construction of the
initial component (LZ77 parsing).

\begin{theorem}
  Given string $T[1..n]$ over alphabet $[1..\sigma]$ of size
  $\sigma\leq n$ encoded in $\bigO(n/\log_{\sigma}n)$ words of space,
  we can compute the Lyndon factorization of $T$ in
  $\bigO(n/\log_{\sigma}n+r \log^9 n + z\log^9 n)$ time and
  $\bigO(n/\log_{\sigma}n + r\log^8 n + z\log^2 n)$ working space.
\end{theorem}
\begin{proof}
  We start by computing the LZ77 parsing using
  Theorem~\ref{thm:lz77-2}.  This takes
  $\bigO(n/\log_{\sigma}n+r\log^9 n + z\log^9 n)$ time and
  $\bigO(n/\log_{\sigma}n+r \log^8 n)$ space. The resulting
  parsing, by Theorem~\ref{thm:lz77-size}, takes
  $\bigO(n/\log_{\sigma}n)$ space.
  
  We then use the Rytter's~\cite{Rytter03} conversion from LZ77 to SLP
  of size $g=\bigO(z \log n)$ that takes $\bigO(z \log n)$ time and
  $\bigO(z \log n)$ working space. The resulting SLP is then turned
  into an LCE data structure of I~\cite{I17}; this takes $\bigO(g \log
  (n/g))=\bigO(z \log^2 n)$ time and $\bigO(g + z\log(n/z))=\bigO(z
  \log n)$ working space. The resulting LCE data structure takes
  $\bigO(z \log n)$ space. Finally, we plug this data structure into
  the algorithm of Furuya~\cite{CPM2018} which gives us the Lyndon
  factorization in $\bigO(z \log^3 n)$ time and $\bigO(z \log^2 n)$
  working space. Thus, the whole pipeline is dominated (in time and
  space) by the construction of LZ77 parsing.
\end{proof}

Similarly as in Section~\ref{sec:lz77}, since $z=\bigO(r \log
n)$~\cite{GNPlatin18}, the above algorithm achieves
$\bigO(n/\log_{\sigma}n)$ runtime and working space when
$n/r\in\Omega({\rm polylog}\,n)$.

\section{Solutions to Textbook Problems}

Lastly, we show how to utilize the techniques presented in this paper
to efficiently solve some ``textbook'' string problems on highly
repetitive inputs. Their solution usually consists of computing $\SA$
or $\LCP$ and performing some simple scan/traversal (e.g., computing
the longest repeating substring amounts to finding the maximal value
in the $\LCP$ array and hence by Theorem~\ref{cor:plcp} it can be
solved efficiently for highly repetitive input), but in some cases
requires explicitly applying some of the observations from previous
sections. Next, we show two examples of such problems.

\subsection{Number of Distinct Substrings}

The number $d$ of distinct substrings of a string $T$ of length $n$ is
given by the formula
\[
  d = \frac{n(n+1)}{2}-\sum_{i=1}^{n}\LCP[i] .
\]

Suppose we are given a (sorted) list $(i_1, \ell_1), \ldots, (i_r,
\ell_r)$ of irreducible lcp values (i.e., $\PLCP[i_k]=\ell_k$) of
string $T$. Since all other lcp values can be derived from this list
using Lemma~\ref{lm:reducible}, we can rewrite the above formula
(letting $i_{r+1}=n+1$) as:
\[
  d = \frac{n(n+1)}{2}-\sum_{k=1}^{r}f(\ell_k, i_{k+1}-i_{k}) ,
\]
\qquad where
\[
    f(v,d) = \left\{
      \begin{array}{ll}
        \frac{v(v+1)}{2} & \text{ if }v < d,\\
        d(v-d)+\frac{d(d+1)}{2} & \text{ otherwise. }
      \end{array}
    \right.
\]

Thus, by Theorem~\ref{cor:plcp} we immediately obtain the following
result.
\begin{theorem}
  Given string $\T[1..n]$ over alphabet $[1..\sigma]$ of size $\sigma
  \leq n$ encoded in $\bigO(n/\log_{\sigma}n)$ words, we can compute
  the number $d$ of distinct substrings of $T$ in
  $\bigO(n/\log_{\sigma} n + r \log^{11} n)$ time and
  $\bigO(n/\log_{\sigma} n + r\log^{10}n)$ space, where $r$ is the
  number of runs in the $\BWT$ of $\T$.
\end{theorem}

\subsection{Longest Substring Occurring \texorpdfstring{$k$}{k} Times}

Suppose that we want to find the length $\ell$ of the longest
substring of $T$ that occurs in $T$ at least $2 \leq k=\bigO(1)$
times. This amounts to computing
\[
  \ell=\max_{i=1}^{n-k+2}\min_{j=0}^{k-2}\LCP[i+j].
\]

For $k=2$ the above formula can be evaluated by only looking at
irreducible lcp values, i.e., using the definition from the previous
section, $\ell=\max_{i=1}^{r}\ell_i$. For $k>2$, this does not work,
since we have to inspect blocks of LCP values of size $k-1$ in
``lex-order''. We instead utilize observations from previous sections.
More precisely, recall from Section~\ref{sec:lz77} that for any
$\tau$, up to a shift by a constant, there is only $r\tau$ different
blocks of size $\tau$ in $\SA$, i.e., for any block block
$\SA[i..i{+}\tau{-}1]$ there exists $k$ such that
$\SA[i..i{+}\tau{-}1]=k+\SA[j..j{+}\tau{-}1]$ and
$\BWT[j..j{+}\tau{-}1]$ contains a BWT-run boundary.

We now observe that an analogous property holds for the $\LCP$ array:
for any block $\LCP[i..i{+}\tau{-}1]$ there exists $k$ (the same as
above) such that $\LCP[i..i{+}\tau{-}1]=\LCP[j..j{+}\tau{-}1]-k$ and
$\BWT[j..j{+}\tau{-}1]$ contains boundary of some BWT-run.  This
implies that we only need to precompute and store the minimum value
inside blocks of $\LCP$ of length $k-1$ that are not further than
$\tau$ positions from the closest BWT-run boundary. All other blocks
of $\LCP$ can be handled using the above observation and the structure
from Section~\ref{sec:rlcsa} for computing the $\LF$-shortcut for any
block of $\BWT$. More precisely, after a suitable overlap (by at least
$k$) of blocks of size $\tau=\Omega({\rm polylog}\,n)$, we can get the
answer for all such blocks in $\bigO(n/{\rm polylog}\,n + r\,{\rm
  polylog}\,n)$ time.

\begin{theorem}
  Given string $\T[1..n]$ over alphabet $[1..\sigma]$ of size
  $\sigma\,{\leq}\, n$ encoded in $\bigO(n/\log_{\sigma}n)$ words, we
  can find the length of the longest substring occurring $\geq $
  $k\,{=}\,\bigO(1)$ times in $T$ in $\bigO(n/\log_{\sigma}n+r\,{\rm
    polylog}\,n)$ time and space.
\end{theorem}

\section{Concluding Remarks}

An important avenue for future work is to reduce the exponent in the
$\bigO(r\,{\rm polylog}\,n)$-term of our bounds and to determine
whether the presented algorithms can be efficiently implemented in
practice.  Another interesting problem is to settle whether the
$\bigO(\log n/\log \log n)$ bound obtained in Section~\ref{sec:rlcsa}
is optimal within $\bigO(r\,{\rm polylog}\,n)$ space.

\section*{Acknowledgments}

We would like to thank Tomasz Kociumaka for helpful comments and
Isamu Furuya, Yuto Nakashima, Tomohiro I, Shunsuke Inenaga, Hideo
Bannai, and Masayuki Takeda for sharing an early version of their
paper~\cite{CPM2018}.

\bibliographystyle{plainurl}
\bibliography{paper}

\end{document}